\documentclass[12pt]{article}
\usepackage{amsmath}
\usepackage{makeidx}
\usepackage{amssymb}
\usepackage{amsfonts}
\usepackage{mitpress}

\newtheorem{theorem}{Theorem}

\newtheorem{lemma}[theorem]{Lemma}

\newtheorem{proposition}[theorem]{Proposition}
\newtheorem{remark}[theorem]{Remark}

\newenvironment{proof}[1][Proof]{\noindent\textbf{#1.} }{\ \rule{0.5em}{0.5em}}
\newdimen\dummy
\dummy=\oddsidemargin
\input{tcilatex}
\begin{document}

\begin{center}
{\large \textbf{\ \\[2mm]
Group of Canonical Diffeomorphisms }}

{\large \textbf{and the Poisson-Vlasov Equations}} \\[1cm]
Hasan G\"{u}mral\\[2mm]

Department of Mathematics, Yeditepe University

34755 Ata\c{s}ehir, \.{I}stanbul, Turkey

hgumral@yeditepe.edu.tr
\end{center}

\bigskip

\bigskip

\textbf{Abstract:} Dynamics of collisionless plasma described by the
Poisson-Vlasov equations is connected with the Hamiltonian motions of
particles and their symmetries. The Poisson equation is obtained as a
constraint arising from the gauge symmetries of particle dynamics.
Variational derivative constrained by the Poisson equation is used to obtain
reduced dynamical equations. Lie-Poisson reduction for the group of
canonical diffeomorphisms gives the momentum-Vlasov equations. Plasma
density is defined as the divergence of symplectic dual of momentum
variables. This definition is also given a momentum map description. An
alternative formulation in momentum variables as a canonical Hamiltonian
system with a quadratic Hamiltonian functional is described. A comparison of
one-dimensional plasma and two-dimensional incompressible fluid is presented.

\newpage

\section{Introduction}

The purpose of this series of papers is to study the geometric structures
underlying both kinematical descriptions and dynamical formulations of
plasma motion, and to provide a geometrical framework for the Poisson-Vlasov
equations

\begin{equation}
\nabla ^{2}\phi _{f}\left( \mathbf{q}\right) =-e\int f\left( \mathbf{z}%
\right) \text{ }d^{3}\mathbf{p}  \label{poi}
\end{equation}

\begin{equation}
\frac{\partial f}{\partial t}+\frac{\mathbf{p}}{m}\cdot \nabla _{q}f-e\nabla
_{q}\phi _{f}\cdot \nabla _{p}f=0  \label{vlasov}
\end{equation}%
of plasma dynamics. These works are initiated from a detailed study of an
unpublished review on plasma dynamics by Marsden and Ratiu \cite{mr93} and
inspired from the problems and ideas in Marsden and Morrison \cite{mm88}. We
refer to \cite{mr94}-\cite{chh98} for background materials and for
references on early works such as \cite{poi90}-\cite{chand77}. The
Poisson-Vlasov system was first written in Hamiltonian form by Morrison \cite%
{mor80}-\cite{mor82}. The Poisson bracket was then shown to be the
Lie-Poisson bracket on the space of plasma densities \cite{mw82}-\cite%
{mwrss83}. 

The kinematical or Lagrangian description will provide us with the
configuration space $G=Diff_{can}(T^{\ast }\mathcal{Q})$; the group of
canonical diffeomorphisms on the phase space $T^{\ast }\mathcal{Q}$ of
motions of individual plasma particles. We shall elaborate, in the next
section, the geometry of $Diff_{can}(T^{\ast }\mathcal{Q})$ which is the
framework for the Hamiltonian (Lie-Poisson) structure of the Poisson-Vlasov
equations.

In section three, we describe the Poisson equation as a kinematical
constraint on the dynamics of Eulerian variables. We show that the Poisson
equation characterizes the set of zero values of the momentum map associated
with the action of additive group of functions $\mathcal{F}(\mathcal{Q})$ on
the position space $\mathcal{Q}$ of particles. This is the gauge group of
particle motion on the canonical phase space $T^{\ast }\mathcal{Q}$.

Momentum map realization of the Poisson equation implies that the true
configuration space for the Poisson-Vlasov dynamics must be the semi-direct
product space $\mathcal{F}(\mathcal{Q})$ $\circledS $ $Diff_{can}(T^{\ast }%
\mathcal{Q})$ with the action of the additive group $\mathcal{F}(\mathcal{Q})
$ of functions given by fiber translation on $T^{\ast }\mathcal{Q}$ and by
composition on right with the canonical transformations. In order to adopt
the configuration space $Diff_{can}(T^{\ast }\mathcal{Q})$, which has been
customary in earlier treatments of the subject \cite{chh98},\cite{mw82}-\cite%
{mwrss83}, we rather proceed by adapting a constraint variational
derivative. In doing so, we implicitly take the advantage of the facts that
the Lie-Poisson structure on $\mathcal{F}^{\ast }(\mathcal{Q})$ is trivial
and that the constraint is of first class. The dual vector space $\mathfrak{g%
}^{\ast }$ of the Lie algebra $\mathfrak{g}\mathbf{=}\mathfrak{X}%
_{ham}(T^{\ast }\mathcal{Q})$ of Hamiltonian vector fields turns out to be
the space of non-closed one-form densities on $T^{\ast }\mathcal{Q}$. By
symplectic duality, this can be identified with the space of non-Hamiltonian
vector fields on $T^{\ast }\mathcal{Q}$. This gives the decomposition of the
tangent space $TT^{\ast }\mathcal{Q}=$ $\mathfrak{g}\mathbf{\oplus }(%
\mathfrak{g}^{\ast })^{\sharp }$ on which one can start a geometric
treatment for Legendre transformation \cite{tulc}. One of the main theme of
the present work is to introduce the dynamical equations on $(\mathfrak{g}%
^{\ast })^{\sharp }-$part of $TT^{\ast }\mathcal{Q}$ and make their
relations with the Vlasov-Poisson equations (\ref{poi}) and (\ref{vlasov})
precise.

In section four, we shall present the kinematical reduction of the dynamics
on $T^{\ast }G$. In the momentum coordinates $(\Pi _{i},\Pi ^{i})$ of $%
\mathfrak{g}^{\ast }$ the Poisson-Vlasov equations take the form%
\begin{equation}
\nabla ^{2}\phi _{\Pi }\left( \mathbf{q}\right) =e\int \nabla _{q}\cdot 
\mathbf{\Pi }_{p}\left( \mathbf{z}\right) d^{3}\mathbf{p}  \label{pipoi}
\end{equation}%
\begin{equation}
\frac{d\Pi _{i}\left( \mathbf{z}\right) }{dt}=-X_{h}(\Pi _{i}\left( \mathbf{z%
}\right) )+e\frac{\partial ^{2}\phi _{f}\left( \mathbf{q}\right) }{\partial
q^{i}\partial q^{j}}\Pi ^{j}\left( \mathbf{z}\right)   \label{piveq}
\end{equation}%
\begin{equation}
\frac{d\Pi ^{i}\left( \mathbf{z}\right) }{dt}=-X_{h}(\Pi ^{i}\left( \mathbf{z%
}\right) )-\frac{1}{m}\delta ^{ij}\Pi _{j}\left( \mathbf{z}\right) 
\label{piueq}
\end{equation}%
and they admit Lie-Poisson Hamiltonian structure with a Hamiltonian function
linear in momenta.

Section five will be devoted to the usual density formulation of the
Poisson-Vlasov equations. The reduced dynamics on $\mathfrak{g}^{\ast }$ has
a further symmetry given by the action of the additive group $\mathcal{F}%
(T^{\ast }\mathcal{Q})$ of functions on $T^{\ast }\mathcal{Q}$. The momentum
map $\mathfrak{g}^{\ast }\rightarrow \mathcal{F}^{\ast }(T^{\ast }\mathcal{Q}%
)=Den(T^{\ast }\mathcal{Q})$ defines the plasma density function $f$. 

In section six, we remark that Eqs.(\ref{piveq}) and (\ref{piueq}) are not
the only equations in momentum variables leading to the Vlasov equation. We 
present another set of equations described by a canonical Hamiltonian
structure with a quadratic Hamiltonian functional.

In section seven, we compare one dimensional plasma with two-dimensional
incompressible fluid. 

\section{Motion of Collisionless Plasma}

\subsection{Kinematical description}

We let $\mathcal{Q}\subset \mathcal{%
\mathbb{R}
}^{3}$ denote the configuration space in which the plasma particles move.
The cotangent bundle $T^{\ast }\mathcal{Q}$ is the corresponding momentum
phase space. This has a natural symplectic structure given by the canonical
two-form $\Omega _{T^{\ast }\mathcal{Q}}=dq^{i}\wedge dp_{i},$ where we
employ the summation over the repeated indices. In the sequel, $\Omega
_{T^{\ast }\mathcal{Q}}^{\sharp }$ will denote the natural isomorphism $%
T^{\ast }T^{\ast }\mathcal{Q}\rightarrow TT^{\ast }\mathcal{Q}$ taking
uniquely a one-form to a vector field on $T^{\ast }\mathcal{Q}$ \cite{mr94},%
\cite{amr88}. We let $t\rightarrow \varphi _{t}$ be a curve such that for
each $t$, $\varphi _{t}$ is a canonical transformation of the particle phase
space $T^{\ast }\mathcal{Q}$ preserving the canonical symplectic two form $%
\Omega _{T^{\ast }\mathcal{Q}}$. Given a point $\mathbf{Z}\in T^{\ast }%
\mathcal{Q}$ regarded as a reference point or, a Lagrangian label, we let $%
\mathbf{z}=\varphi _{t}(\mathbf{Z})$, also written as $\mathbf{z}=(\mathbf{q}%
,\mathbf{p})=\varphi (\mathbf{Z},t)$, denote the current phase space point
or, the Eulerian coordinates of plasma particles. The phase space velocity
is given by the time dependent vector 
\begin{equation}
\mathbf{\dot{z}}=\frac{d}{dt}\varphi _{t}(\mathbf{Z})=X_{t}(\varphi _{t}(%
\mathbf{Z}))=X(\mathbf{z},t)
\end{equation}%
and it generates the flow $\varphi _{t}$. Since $\varphi _{t}$ is canonical, 
$X$ is infinitesimally Hamiltonian. We shall assume that it is globally
Hamiltonian and write $h(\mathbf{z},t)$ for the corresponding Hamiltonian
function so that at each $t$, $X=X_{h}$. 

\subsection{\label{group}Group of canonical diffeomorphisms}

The flow $\varphi _{t}$ of the Hamiltonian vector field $X_{h}$ on $T^{\ast }%
\mathcal{Q}$ is a one parameter family of elements of the group $%
G=Diff_{can}(T^{\ast }\mathcal{Q})$ of all transformations of $T^{\ast }%
\mathcal{Q}$ preserving the symplectic two-form $\Omega _{T^{\ast }\mathcal{Q%
}}$. In this work, we restrict the discussion to the canonical
transformations connected to the identity \cite{mr94},\cite{leo01}.$\
\varphi _{t}$ acts on left by evaluation on the space $T^{\ast }\mathcal{Q}$
of reference plasma configuration to produce the motion of particles. Thus,
a configuration of plasma can be specified by an element of \ $G$. The right
action of $G$ commutes with the particle motions and constitute an infinite
dimensional symmetry group of the kinematical description. This is the
particle relabelling symmetry \cite{ak}. For the motion of particles on $%
T^{\ast }\mathcal{Q}$ described by the left action 
\begin{equation}
G\times T^{\ast }\mathcal{Q}\longrightarrow T^{\ast }\mathcal{Q}:(\varphi ,%
\mathbf{Z})\rightarrow L_{\varphi }(\mathbf{Z})=\varphi (\mathbf{Z})=\mathbf{%
z}
\end{equation}%
the velocity field on $T^{\ast }\mathcal{Q}$ is the vector $X_{\varphi }$
lying in the tangent space 
\begin{equation}
T_{\varphi }G=\{X_{\varphi }:T^{\ast }\mathcal{Q}\rightarrow TT^{\ast }%
\mathcal{Q}|\;\;\tau _{T^{\ast }\mathcal{Q}}\circ X_{\varphi }=\varphi \}
\end{equation}%
of $G$ at $\varphi $. Here, $\tau _{T^{\ast }\mathcal{Q}}:TT^{\ast }\mathcal{%
Q}\rightarrow T^{\ast }\mathcal{Q}$ is the natural projection of the tangent
bundle of particle phase space. That means, $X_{\varphi }$ is a map over the
element $\varphi :T^{\ast }\mathcal{Q}\rightarrow T^{\ast }\mathcal{Q}$ of
the configuration space. Vector fields on $G$ are then defined to be the
maps $X:G\rightarrow TG$ whose value at $\varphi \in G$ is given by $%
X_{\varphi }\in T_{\varphi }G$. Since the motion is associated with the left
action, the velocity field $X$ is invariant under the right action $R_{\psi
}:G\rightarrow G$ of $G$ for all $\psi \in G$.

The space $\{X|\;(R_{\psi })_{\ast }X=TR_{\psi }\circ X\mathbf{\circ }\psi
^{-1}=X,\;\forall \psi \in G\}$ of right invariant vector fields on $G$ is
isomorphic to the tangent space over the identity mapping of $T^{\ast }%
\mathcal{Q}$ with the isomorphism given at each $\varphi \in G$ by 
\begin{equation}
X_{h}\rightarrow X_{h}\circ \varphi \;\;\;\forall \;X_{h}\in T_{id}G
\label{e-l}
\end{equation}%
where $X_{h}$ is the Hamiltonian vector field $\Omega _{T^{\ast }\mathcal{Q}%
}^{\sharp }(dh)$ on $T^{\ast }\mathcal{Q}$ whose flow is $\varphi $. Then,
the Lie algebra 
\begin{equation}
\mathfrak{g}=(\mathfrak{X}_{ham}(T^{\ast }\mathcal{Q})\;;\;-[\;,\;])
\end{equation}%
of $G$ consists of Hamiltonian vector fields on $T^{\ast }\mathcal{Q}$ and $%
[\;,\;]$ denotes the standard Jacobi-Lie bracket, with conventions as in 
\cite{amr88}. The isomorphism in Eq.(\ref{e-l}) is just the relation 
\begin{equation}
X_{\varphi }(\mathbf{Z},t)=X_{h}(\mathbf{z},t)=(X_{h}\circ \varphi _{t})(%
\mathbf{Z})
\end{equation}%
between the Lagrangian and the Eulerian velocities at the point $\mathbf{z}$%
. In coordinates, if we decompose a canonical transformation $\varphi $ into
position-momentum pairs as $\varphi =(\xi ,\eta )$, then, we have 
\begin{equation}
X_{\varphi }(\mathbf{Z},t)=\dot{\xi}^{i}\left( \mathbf{Z}\right) \dfrac{%
\partial }{\partial q^{i}}+\dot{\eta}_{i}\left( \mathbf{Z}\right) \dfrac{%
\partial }{\partial p_{i}}=\dot{q}^{i}\dfrac{\partial }{\partial q^{i}}+\dot{%
p}_{i}\dfrac{\partial }{\partial p_{i}}=X_{h}(\mathbf{z},t).
\end{equation}

By the identification $X_{h}\rightarrow h$ modulo constants, and the
homomorphism 
\begin{equation}
\lbrack X_{h},X_{k}]=-X_{\{h,k\}_{T^{\ast }\mathcal{Q}}},
\end{equation}%
we have the identification $\mathfrak{g}\simeq (\mathcal{F}(T^{\ast }%
\mathcal{Q});\{\;,\;\}_{T^{\ast }\mathcal{Q}})$ of the Lie algebra with the
space of functions on $T^{\ast }\mathcal{Q}$ endowed with the canonical
Poisson bracket.

We define the dual vector space $\mathfrak{g}^{\ast }$ of the Lie algebra $%
\mathfrak{g}\mathbf{=}\mathfrak{X}_{ham}(T^{\ast }\mathcal{Q})$ to be the
non-closed one-form densities on $T^{\ast }\mathcal{Q}$ 
\begin{equation}
\mathfrak{g}^{\ast }=\{\Pi _{id}\otimes d\mu \in \Lambda ^{1}(T^{\ast }%
\mathcal{Q})\otimes Den(T^{\ast }\mathcal{Q}):d\Pi _{id}\neq 0\}  \label{gst}
\end{equation}%
where $d\mu \in Den(T^{\ast }\mathcal{Q})$ is a volume six-form on $T^{\ast }%
\mathcal{Q}$. At each point $\mathbf{z\in }$ $T^{\ast }\mathcal{Q}$ the
space $\Lambda ^{6}(T^{\ast }\mathcal{Q})$ of six-forms is one-dimensional
whose basis can be chosen to be the symplectic volume $d\mu =\Omega
_{T^{\ast }\mathcal{Q}}^{3}=\Omega _{T^{\ast }\mathcal{Q}}\wedge \Omega
_{T^{\ast }\mathcal{Q}}\wedge \Omega _{T^{\ast }\mathcal{Q}}$. With these
definitions we have the non-degenerate pairing 
\begin{eqnarray}
\langle X_{h},\Pi _{id}\otimes d\mu \rangle  &=&\int_{T^{\ast }\mathcal{Q}%
}\;<X_{h}(\mathbf{z}),\Pi _{id}(\mathbf{z})>\;d\mu (\mathbf{z})  \notag \\
&=&\int_{T^{\ast }\mathcal{Q}}\;h(\mathbf{z})\,\nabla _{z}\cdot \mathbf{\Pi }%
_{id}^{\sharp }(\mathbf{z})\text{ }d\mu (\mathbf{z})\text{ \ \ \ }  \notag \\
&=&\langle h,\,\nabla _{z}\cdot \mathbf{\Pi }_{id}^{\sharp }\otimes d\mu
\rangle   \label{pair}
\end{eqnarray}%
of $\mathfrak{g}$ and $\mathfrak{g}^{\ast }$ with respect to the $L^{2}$%
-norm. Here, $\mathbf{\Pi }_{id}^{\sharp }$ denotes the components of the
vector $\Pi _{id}^{\sharp }=\Omega _{T^{\ast }\mathcal{Q}}^{\sharp }(\Pi
_{id})=\mathbf{\Pi }_{id}^{\sharp }(\mathbf{z})\cdot \nabla _{z}$ and the
pairing in the integrand is defined over the finite dimensional space $%
T^{\ast }\mathcal{Q}$. From the last line of the above equations we conclude
that the pairing of algebra with its dual is nondegenerate if $\nabla
_{z}\cdot \mathbf{\Pi }_{id}^{\sharp }(\mathbf{z})\neq 0$. Since, $\Omega
_{T^{\ast }\mathcal{Q}}$ is nondegenerate, this is equivalent to the
condition $d\Pi _{id}\neq 0$. The definition of dual space $\mathfrak{g}%
^{\ast }$ of the Lie algebra $\mathfrak{g}$ already implies that the
identification of vector space $\mathfrak{g}=\mathfrak{X}_{ham}(T^{\ast }%
\mathcal{Q})$ with the space of functions $\mathcal{F(}T^{\ast }\mathcal{Q)}$
can be extended to the identification of dual space with the space $%
Den(T^{\ast }\mathcal{Q})$ of densities on $T^{\ast }\mathcal{Q}$ via the
definition of density

\begin{equation}
f\left( \mathbf{z}\right) =\nabla _{z}\cdot \mathbf{\Pi }_{id}^{\sharp }(%
\mathbf{z})=\frac{\partial \Pi _{i}\left( \mathbf{z}\right) }{\partial p_{i}}%
-\frac{\partial \Pi ^{i}\left( \mathbf{z}\right) }{\partial q^{i}}\neq 0
\label{density}
\end{equation}%
which does not vanish due to nondegeneracy restriciton in the definition of $%
\mathfrak{g}^{\ast }$. That is, we have the one-sided correspondence

\begin{equation*}
\Pi _{id}\otimes d\mu \in \mathfrak{g}^{\ast }\longrightarrow f\,\text{\ }%
d\mu \in Den(T^{\ast }\mathcal{Q})
\end{equation*}%
with the definition of $f$ given by Eq.(\ref{density}) (see also the
internet supplement to \cite{mr94}). We shall show later the precise way of
obtaining density in the context of a momentum map. By the above
identifications with function spaces the pairing between Lie algebra $%
\mathfrak{g}\mathbf{\equiv }\mathcal{F(}T^{\ast }\mathcal{Q)}$ and its dual $%
\mathfrak{g}^{\ast }\equiv Den(T^{\ast }\mathcal{Q})$ takes the form of
multiply-and-integrate 
\begin{equation}
\langle h,f\otimes d\mu \rangle =\int_{T^{\ast }Q}h(\mathbf{z})f(\mathbf{z})%
\text{ }d\mu (\mathbf{z})\text{ .}  \label{pairfh}
\end{equation}

\begin{remark}
The configuration space $Diff_{can}(T^{\ast }\mathcal{Q})$ can also be given
a description in terms of sections of the trivial bundle $T^{\ast }\mathcal{Q%
}_{0}\mathcal{\times }T^{\ast }\mathcal{Q\rightarrow }T^{\ast }\mathcal{Q}%
_{0}$ where $T^{\ast }\mathcal{Q}_{0}$ is the particle phase space with
Lagrangian coordinates $\mathbf{Z}$ and $T^{\ast }\mathcal{Q}$ carries the
Eulerian coordinates $\mathbf{z}$. The total space $T^{\ast }\mathcal{Q}_{0}%
\mathcal{\times }T^{\ast }\mathcal{Q}$ is symplectic with the two-form $%
\Omega _{-}=\Omega _{T^{\ast }\mathcal{Q}_{0}}-\Omega _{T^{\ast }\mathcal{Q}%
} $. A diffeomorphism $\varphi :T^{\ast }\mathcal{Q}_{0}\mathcal{\rightarrow 
}T^{\ast }\mathcal{Q}$ is canonical if $\Omega _{T^{\ast }\mathcal{Q}%
_{0}}-\varphi ^{\ast }\Omega _{T^{\ast }\mathcal{Q}}=0$. It follows that $%
\Omega _{-}$ vanishes when restricted to the graphs of canonical
diffeomorphisms. Graphs are elements of the space $\Gamma (T^{\ast }\mathcal{%
Q}_{0}\mathcal{\times }T^{\ast }\mathcal{Q)}$ of sections of the trivial
bundle $T^{\ast }\mathcal{Q}_{0}\mathcal{\times }T^{\ast }\mathcal{%
Q\rightarrow }T^{\ast }\mathcal{Q}_{0}$. For a base point $\mathbf{Z\in }%
T^{\ast }\mathcal{Q}_{0}$, the total space is twelve dimensional and the
graph $\left( \mathbf{Z},\varphi \left( \mathbf{Z}\right) \right) $ of a
diffeomorphism is a six dimensional subspace. When $\varphi $ is canonical, $%
\Omega _{-}$ vanishes on graphs and such graphs are called Lagrangian
subspaces. If we denote the space of all sections of the trivial bundle on
which the restriction of $\Omega _{-}$ vanishes, namely the space of all
Lagrangian sections, by $Lag(\Gamma (T^{\ast }\mathcal{Q}_{0}\mathcal{\times 
}T^{\ast }\mathcal{Q))}$, then we have the identification 
\begin{equation*}
Diff_{can}(T^{\ast }\mathcal{Q})\simeq Lag(\Gamma (T^{\ast }\mathcal{Q}_{0}%
\mathcal{\times }T^{\ast }\mathcal{Q))}\text{ .}
\end{equation*}%
To have the corresponding description for the Lie algebra of Hamiltonian
vector fields, first observe that we can identify the dual space $\mathfrak{g%
}^{\ast }$ with the space of vector fields 
\begin{equation*}
(\mathfrak{g}^{\ast })^{\sharp }=\{\;\Pi _{id}^{\sharp }=\Omega _{T^{\ast }%
\mathcal{Q}}^{\sharp }(\Pi _{id})\in TT^{\ast }\mathcal{Q}\;|\;d\Pi
_{id}\neq 0\;\}
\end{equation*}%
which are not Hamiltonian by definition. Then, we have the decomposition of
the tangent space $TT^{\ast }\mathcal{Q}=$ $\mathfrak{g}\mathbf{\oplus }(%
\mathfrak{g}^{\ast })^{\sharp }$ into the underlying vector spaces of the
Lie algebra and its dual. It turns out that the Lie algebra of Hamiltonian
vector fields can be identified with the space of all Lagrangian
submanifolds of $\Gamma (TT^{\ast }\mathcal{Q)}$ with respect to the
Tulczyjew symplectic structure on $TT^{\ast }\mathcal{Q}$ \cite{gpd2}.
\end{remark}

To extend the pairing of $\mathfrak{g}$ and $\mathfrak{g}^{\ast }$ to a
pairing of tangent and cotangent bundles we define an element $\Pi _{\varphi
}$ of the covector space $\Lambda _{\varphi }^{1}G$ of one-forms as a map
over $\varphi $ 
\begin{equation}
\Lambda _{\varphi }^{1}G=\{\Pi _{\varphi }:T^{\ast }\mathcal{Q}\mapsto
T^{\ast }T^{\ast }\mathcal{Q}\;,\;\pi _{T^{\ast }\mathcal{Q}}\circ \Pi
_{\varphi }=\varphi \}
\end{equation}%
where $\pi _{T^{\ast }\mathcal{Q}}:T^{\ast }T^{\ast }\mathcal{Q\rightarrow }%
T^{\ast }\mathcal{Q}$ is the natural projection, and a 1-form field $\Pi $
on $G$ to be a section $\Pi :G\longrightarrow \Lambda ^{1}G$ such that for
each $\varphi \in G$ we have $\Pi _{\varphi }\in \Lambda _{\varphi }^{1}G$.
The cotangent space at $\varphi $ is then the space 
\begin{equation}
T_{\varphi }^{\ast }G=\{\Pi _{\varphi }\otimes d\mu \in \Lambda _{\varphi
}^{1}G\otimes Den(T^{\ast }\mathcal{Q})\}
\end{equation}%
of one-form densities on $G=Diff_{can}(T^{\ast }\mathcal{Q)}$. The pairing
of tangent and cotangent spaces at $\varphi \in G$ becomes 
\begin{equation}
\langle X_{\varphi },\Pi _{\varphi }\otimes d\mu \rangle =\int_{T^{\ast }%
\mathcal{Q}}\;<X_{\varphi }(\mathbf{Z}),\Pi _{\varphi }(\mathbf{Z})>\;d\mu (%
\mathbf{Z})\;.
\end{equation}%
Following diagram summarizes the mapping properties with reference to
particle phase space of elements of $T_{\varphi }G$, $T_{\varphi }^{\ast }G$%
, $\mathfrak{g}$ and $\mathfrak{g}^{\ast }$

\begin{equation*}
\begin{array}{ccccc}
& TT^{\ast }\mathcal{Q}\text{ \ \ } & \underleftrightarrow{\Omega _{T^{\ast }%
\mathcal{Q}}^{\flat },\text{ }\Omega _{T^{\ast }\mathcal{Q}}^{\sharp }} & 
\text{\ }T^{\ast }T^{\ast }\mathcal{Q} &  \\ 
\begin{array}{c}
\begin{array}{cc}
X_{\varphi } & \text{ \ }{\huge \nearrow }\text{ \ \ }%
\end{array}
\\ 
\mathstrut%
\end{array}
& 
\begin{array}{c}
\begin{array}{cc}
X_{h}{\Huge \uparrow } & {\Huge \downarrow }\tau _{T^{\ast }\mathcal{Q}}%
\end{array}
\\ 
\end{array}
&  & 
\begin{array}{c}
\begin{array}{cc}
\pi _{T^{\ast }\mathcal{Q}}{\Huge \downarrow } & {\Huge \uparrow }\Pi _{id}%
\end{array}
\\ 
\end{array}
& 
\begin{array}{c}
\begin{array}{cc}
\text{ \ \ \ }{\huge \nwarrow }\text{ \ } & \Pi _{\varphi }%
\end{array}
\\ 
\mathstrut%
\end{array}
\\ 
T^{\ast }\mathcal{Q}\text{ \ \ \ \ }\overrightarrow{\text{ \ \ }\varphi 
\text{ \ \ }} & T^{\ast }\mathcal{Q} & \equiv & T^{\ast }\mathcal{Q} & 
\overleftarrow{\text{ \ \ }\varphi \text{ \ \ }}\text{\ \ \ \ }T^{\ast }%
\mathcal{Q}%
\end{array}%
\text{ .}
\end{equation*}

The adjoint action $Ad_{\varphi }:\mathfrak{g}\rightarrow \mathfrak{g}$ of $G
$ on its Lie algebra is given by push forward of Hamiltonian vector fields
on $T^{\ast }\mathcal{Q}$ 
\begin{equation}
Ad_{\varphi }(X_{h})=T\varphi \circ X_{h}\circ \varphi ^{-1}=\varphi _{\ast
}(X_{h})\;.
\end{equation}%
The coadjoint action $(Ad_{\varphi })^{\ast }:\mathfrak{g}^{\ast
}\rightarrow \mathfrak{g}^{\ast }$ of $G$ on $\mathfrak{g}^{\ast }$ is the
dual $Ad_{\varphi ^{-1}}^{\ast }$ with respect to the pairing in Eq.(\ref%
{pair}) of the map $Ad_{\varphi ^{-1}}$ and is given by push forward as well

\begin{equation}
Ad_{\varphi ^{-1}}^{\ast }\left( \Pi _{id}\right) =\varphi _{\ast }\left(
\Pi _{id}\right) =T_{id}^{\ast }(R_{\varphi }\circ L_{\varphi ^{-1}})\circ
\Pi _{id}\text{ .}
\end{equation}%
Taking the derivatives at the identity one finds that the adjoint action $%
ad_{X_{h}}$ of $\mathfrak{g}$ on itself and the coadjoint action $%
ad_{X_{h}}^{\ast }$ of $\mathfrak{g}$ on $\mathfrak{g}^{\ast }$ are
generated by the Lie derivative with respect to the Hamiltonian vector field 
$X_{h}\in T$ $T^{\ast }\mathcal{Q}$. In other words, given $X_{h},X_{k}\in 
\mathfrak{g}$ and $\Pi _{id}\in \mathfrak{g}^{\ast }$, the Lie derivative $%
\mathcal{L}_{X_{h}}$ may be regarded as tangent vectors $ad_{X_{h}}\in
T_{X_{k}}\mathfrak{g}$ or $ad_{X_{h}}^{\ast }\in T_{\Pi _{id}}\mathfrak{g}%
^{\ast }$ depending on whether it is associated with adjoint or coadjoint
actions, respectively.

The cotangent lift to $T_{\psi }^{\ast }G$ of the right action is $%
T_{\varphi \circ \psi ^{-1}}^{\ast }R_{\psi }(\Pi _{\varphi })$ $=$ $\Pi
_{\varphi }\circ \psi ^{-1}.$ For $\psi =\varphi $ this gives the
translation of the one form $\Pi _{\varphi }$ at $\varphi $ to the identity
which we denote by $\Pi _{id}\equiv \Pi _{\varphi }\circ \varphi ^{-1}.$ By
definition, this is an element of the dual space $\mathfrak{g}^{\ast }$ of
the Lie algebra $\mathfrak{g}$ of Hamiltonian vector fields. The right
invariant momentum map $\mathbb{J}_{L}:T^{\ast }G\rightarrow \mathfrak{g}%
^{\ast }$ for the lifted left action (i.e. the plasma motion) is defined by 
\begin{eqnarray}
\left\langle \mathbb{J}_{L}\left( \Pi _{\varphi }\right) \otimes d\mu
,X_{h}\right\rangle  &=&\left\langle \Pi _{\varphi }\otimes d\mu
,T_{id}R_{\varphi }\circ X_{h}\right\rangle   \notag \\
&=&\left\langle \Pi _{id}\otimes d\mu ,X_{h}\right\rangle 
\end{eqnarray}%
so that we have $\mathbb{J}_{L}(\Pi _{\varphi })=\Pi _{id}\in \mathfrak{g}%
^{\ast }$ \cite{mr94}.

\section{Poisson Equation as a Constraint}

\subsection{Momentum map description of Poisson equation}

The canonical symplectic structure on $T^{\ast }\mathcal{Q}$ is invariant
under translation of fiber variable by an exact one-form over $\mathcal{Q}.$
This is the gauge transformation of canonical Hamiltonian formalism. If we
identify $T^{\ast }\mathcal{Q}$ with the space of one-forms $\Lambda ^{1}(%
\mathcal{Q})$ on $\mathcal{Q}$, this invariance may be described as the
Hamiltonian action on $T^{\ast }\mathcal{Q}$ 
\begin{equation}
\Lambda ^{0}(\mathcal{Q})\times \Lambda ^{1}(\mathcal{Q})\rightarrow \Lambda
^{1}(\mathcal{Q}):\left( \phi \left( \mathbf{q}\right) ,\mathbf{p\cdot }d%
\mathbf{q}\right) \rightarrow \mathbf{p\cdot }d\mathbf{q}+d\phi \left( 
\mathbf{q}\right) 
\end{equation}%
of the space $\Lambda ^{0}(\mathcal{Q})\equiv \mathcal{F}(\mathcal{Q})$ of
zero-forms$.$ From an algebraic point of view, the exterior derivative $%
d:\Lambda ^{0}(\mathcal{Q})\rightarrow \Lambda ^{1}(\mathcal{Q})$ can be
interpreted as a map describing a Lie algebra isomorphism of the additive
algebra of functions $\mathcal{F}(\mathcal{Q})$ into the additive algebra of
one-forms $\Lambda ^{1}(\mathcal{Q})$ \cite{gs}. As the dual of any Lie
algebra isomorphism is a momentum map, we have 
\begin{equation}
\mathbb{J}_{\mathcal{F}(\mathcal{Q})}:\Lambda ^{2}(\mathcal{Q})\rightarrow
Den(\mathcal{Q})
\end{equation}%
where the space of two forms $\Lambda ^{2}(\mathcal{Q})$ on $\mathcal{Q}$
and the space of densities (three-forms) on $\mathcal{Q}$ are the duals,
with respect to the $L^{2}-$norm, of $\Lambda ^{1}(\mathcal{Q})$ and $%
\Lambda ^{0}(\mathcal{Q})$, respectively. We can identify $\Lambda ^{1}(%
\mathcal{Q})$ and its dual space $\Lambda ^{2}(\mathcal{Q})$ by the Hodge
duality oparator $\ast $ associated with a Riemannian metric on $\mathcal{Q}$%
. Then 
\begin{equation}
\left\langle \mathbb{J}_{\mathcal{F}(\mathcal{Q})}\left( \mathbf{p\cdot }d%
\mathbf{q},\ast d\phi \left( \mathbf{q}\right) \right) ,\phi \left( \mathbf{q%
}\right) \right\rangle =-\int_{\mathcal{Q}}\phi \left( \mathbf{q}\right)
d\ast d\phi \left( \mathbf{q}\right) 
\end{equation}%
where $\ast d\phi \left( \mathbf{q}\right) $ is considered as a two-form
over the one-form $\mathbf{p\cdot }d\mathbf{q}$, and we used the pairing of $%
\Lambda ^{1}(\mathcal{Q})$ and $\Lambda ^{2}(\mathcal{Q})$ given by
integration over $\mathcal{Q}$. Thus, the momentum map is 
\begin{equation*}
\mathbb{J}_{\mathcal{F}(\mathcal{Q})}\left( \mathbf{p\cdot }d\mathbf{q},\ast
d\phi \left( \mathbf{q}\right) \right) =-d\ast d\phi \left( \mathbf{q}%
\right) \in Den(\mathcal{Q})\text{ .}
\end{equation*}%
For the Euclidean metric on $\mathcal{Q}$, the operator $d\ast d$ is the
usual Laplacian $\nabla _{q}^{2}$ in Cartesian coordinates.

The action of the additive group $\mathcal{F}(\mathcal{Q})$ of functions can
be carried over objects defined on $T^{\ast }\mathcal{Q}$ such as the group
of canonical diffeomorphisms, its Lie algebra and the dual of the Lie
algebra. For our purpose of obtaining the Poisson equation as a constraint
described by a momentum map, we restrict ourselves to the action of $%
\mathcal{F}(\mathcal{Q})$ on the space $\mathcal{F}(T^{\ast }\mathcal{Q})$
of functions on $T^{\ast }\mathcal{Q}$. It will be convenient to describe
the momentum map as the dual of some Lie algebra isomorphism into. We think
of $\mathcal{F}(T^{\ast }\mathcal{Q})$ equipped with the Poisson bracket to
be an algebra isomorphic to the Lie algebra of Hamiltonian vector fields
(c.f. section \ref{group}). Then, $\mathcal{F}(\mathcal{Q})$ is a
commutative subalgebra of $(\mathcal{F}(T^{\ast }\mathcal{Q}%
),\{\;,\;\}_{T^{\ast }\mathcal{Q}})$ corresponding to the generators of the
action on $T^{\ast }\mathcal{Q}$ by fiber translation $\mathbf{p}\mapsto 
\mathbf{p}-\nabla _{q}\phi \left( \mathbf{q}\right) $. The Hamiltonian
function is $-\phi \left( \mathbf{q}\right) \in \mathcal{F}(\mathcal{Q}%
)\subset \mathcal{F}(T^{\ast }\mathcal{Q})$. Thus, the required Lie algebra
isomorphism is from the additive algebra of functions $\mathcal{F}(\mathcal{Q%
})$ into the Poisson bracket algebra on $\mathcal{F}(T^{\ast }\mathcal{Q})$.
Together with the dualization we have%
\begin{equation*}
\begin{array}{ccc}
(\mathcal{F}(\mathcal{Q}),+) & 
\begin{array}{c}
\underrightarrow{\text{ \ \ }\ \ \text{ \ \ }} \\ 
\mathstrut%
\end{array}
& (\mathcal{F}(T^{\ast }\mathcal{Q}),\{\;,\;\}_{T^{\ast }\mathcal{Q}}) \\ 
\updownarrow &  & \updownarrow \\ 
Den(\mathcal{Q}) & 
\begin{array}{c}
\underleftarrow{\text{ \ }\mathbb{J}_{\mathcal{F}(\mathcal{Q})}\text{ \ }}
\\ 
\mathstrut%
\end{array}
& Den(T^{\ast }\mathcal{Q})%
\end{array}%
\end{equation*}%
and the momentum map $\mathbb{J}_{\mathcal{F}(\mathcal{Q})}:Den(T^{\ast }%
\mathcal{Q})\rightarrow Den(\mathcal{Q})$ is computed from 
\begin{equation}
\left\langle \mathbb{J}_{\mathcal{F}(\mathcal{Q})}\left( f\left( \mathbf{z}%
\right) d\mu \left( \mathbf{z}\right) \right) ,\phi \left( \mathbf{q}\right)
\right\rangle =-\int_{T^{\ast }\mathcal{Q}}f\left( \mathbf{z}\right) \phi
\left( \mathbf{q}\right) d\mu \left( \mathbf{z}\right)
\end{equation}%
to be the volume density 
\begin{equation}
\mathbb{J}_{\mathcal{F}(\mathcal{Q})}\left( f\left( \mathbf{z}\right) d\mu
\left( \mathbf{z}\right) \right) =-(\int_{T^{\ast }\mathcal{Q}}f\left( 
\mathbf{z}\right) \text{ }d^{3}\mathbf{p)}\text{ }d^{3}\mathbf{q}
\end{equation}%
on $\mathcal{Q}$ \cite{mar82}$.$ Combining the actions of $\mathcal{F}(%
\mathcal{Q})$ on $\Lambda ^{1}(\mathcal{Q})$ and $\mathcal{F}(T^{\ast }%
\mathcal{Q})$ we have the momentum map%
\begin{equation*}
\mathbb{J}_{\mathcal{F}(\mathcal{Q})}:\Lambda ^{2}(\mathcal{Q})\times
Den(T^{\ast }\mathcal{Q})\longrightarrow Den(\mathcal{Q})
\end{equation*}%
given by%
\begin{equation*}
\mathbb{J}_{\mathcal{F}(\mathcal{Q})}\left( \ast d\phi \left( \mathbf{q}%
\right) ,ef\left( \mathbf{z}\right) d\mu \left( \mathbf{z}\right) ;\phi
\left( \mathbf{q}\right) \right) =-(\nabla ^{2}\phi \left( \mathbf{q}\right)
+e\int f\left( \mathbf{z}\right) \text{ }d^{3}\mathbf{p})\text{ }d^{3}%
\mathbf{q}
\end{equation*}%
whose zero value is the Poisson equation. This constraints the region in the
product space $\Lambda ^{2}(\mathcal{Q})\times Den(T^{\ast }\mathcal{Q})$
for consideration of dynamics in the Eulerian variables $(\phi ,f)$, namely%
\begin{equation*}
\begin{array}{r}
\mathbb{J}_{\mathcal{F}(\mathcal{Q})}^{-1}(0)/\mathcal{F}(\mathcal{Q}%
)=\{(\ast d\phi ,fd\mu )\in \Lambda ^{2}(\mathcal{Q})\times Den(T^{\ast }%
\mathcal{Q})\mid \\ 
\nabla ^{2}\phi \left( \mathbf{q}\right) +e\int f\left( \mathbf{z}\right) 
\text{ }d^{3}\mathbf{p=}0\}%
\end{array}%
\end{equation*}%
is the reduced space for the Lie-Poisson description. This corresponds to a
subset of the dual space $\Lambda ^{2}(\mathcal{Q})\times \mathfrak{g}^{\ast
}$ of the Lie algebra $\Lambda ^{1}(\mathcal{Q})\times \mathfrak{g}$ with
trivial bracket on the first factor. The underlying Lie group is $\mathcal{F}%
(\mathcal{Q})\times Diff_{can}(T^{\ast }\mathcal{Q})$ with the first factor
acting on canonical diffeomorphisms by composition with fiber translations.

\subsection{Constraint variational derivative}

The momentum map description of Poisson equation implies that we have to
consider the configuration variables of collisionless plasma motion to be $%
\left( \phi ,\varphi \right) $ where the electrostatic potential $\phi $ is
a function on $\mathcal{Q}$ and $\varphi $ is a canonical diffeomorphism of $%
T^{\ast }\mathcal{Q}$ generating the particle motion. Hence, the
configuration space of plasma motion must be $\mathcal{F}(\mathcal{Q}%
)\circledS Diff_{can}(T^{\ast }\mathcal{Q})$ where $\circledS $ denotes the
semidirect product of groups with the additive group $\mathcal{F}(\mathcal{Q}%
)$ of functions acting on the second factor by composition on right. Here,
we want to adopt an approach allowing the possibility to use much simpler
configuration space $G=Diff_{can}(T^{\ast }\mathcal{Q})$. More precisely, we
want to use the Poisson equation as a constraint for the variational
derivatives of Eulerian variables, in particular, the plasma density
function.

To this end, we consider the Green's function solution 
\begin{equation}
\phi _{f}(q,t)=e\int_{T^{\ast }\mathcal{Q}}\;K(\mathbf{q}|\mathbf{\acute{q}}%
)f(\mathbf{\acute{z}})\;d\mu (\mathbf{\acute{z}})  \label{greep}
\end{equation}%
of the Poisson equation (\ref{poi}) which relates the plasma density $f$ and
the electrostatic potential $\phi _{f}$ \cite{mor80},\cite{mor81},\cite{kd84}%
. As an example of Eulerian quantities, we take the Hamiltonian function

\begin{equation}
H_{LP}(f)=\int_{T^{\ast }Q}f(\mathbf{z})h_{f}\left( \mathbf{z}\right) d\mu (%
\mathbf{z})  \label{piham}
\end{equation}%
of the Lie-Poisson formulation \cite{mor80},\cite{mor81},\cite{mar82},\cite%
{kd84}. Here, the density dependent function%
\begin{equation}
h_{f}\left( \mathbf{z}\right) =\frac{p^{2}}{2m}+\frac{1}{2}e\phi _{f}(%
\mathbf{q})  \label{hf}
\end{equation}%
is related to the Hamiltonian function $h$ governing the particle dynamics
up to a multiplicative factor in potential term. The Hamiltonian functional
in Eq.(\ref{piham}) is the total energy of the plasma in Eulerian
coordinates.

\begin{lemma}
\cite{mor81} For the functional in Eq.(\ref{piham}) with Eq.(\ref{hf}) we
have%
\begin{equation*}
{\frac{\delta H_{LP}(f)}{\delta f}}={\frac{1}{2m}}p^{2}+e\phi _{f}(\mathbf{q}%
)=h(\mathbf{z}).
\end{equation*}

\begin{proof}
Using the Poisson equation and the Green`s function solution, $H_{LP}(f)$
can be put into the form 
\begin{eqnarray}
H_{LP}(f) &=&\int_{T^{\ast }\mathcal{Q}}\;{\frac{1}{2m}}p^{2}f(\mathbf{z}%
)d\mu (\mathbf{z})  \notag \\
&&+{\frac{e^{2}}{2}}\int_{T^{\ast }\mathcal{Q}}\;\int_{T^{\ast }\mathcal{Q}%
}f(\mathbf{z})K(\mathbf{q}|\mathbf{\acute{q}})f(\mathbf{\acute{z}})\;d\mu (%
\mathbf{z})\;d\mu (\mathbf{\acute{z}})  \label{ede}
\end{eqnarray}%
up to the integral of the divergence term $\nabla _{q}\cdot (\phi _{f}(%
\mathbf{q})\nabla _{q}\phi _{f}(\mathbf{q}))$. It is now easy to obtain the
lemma where a factor of $2$ comes from the symmetry of the Green's function 
\cite{mor80},\cite{mor81},\cite{kd84}. 
\end{proof}
\end{lemma}

The constraint imposed by the Poisson equation is, in the language of Dirac
formalism, first class and hence does not affect the Poisson bracket on the
reduced space \cite{agw89}. Thus, in obtaining equivalent dynamical
formulations in alternative Eulerian variables we must use the same
constraint. The foremost example of such a variable arises from the
identification of the dual space $\mathfrak{g}^{\ast }$ of the algebra of
Hamiltonian vector fields with the space of densities $Den(T^{\ast }\mathcal{%
Q})$. In the more basic formulation of dynamics with the momentum variables $%
\Pi _{id}\in $ $\mathfrak{g}^{\ast }$ the Hamiltonian functional turns out
to be 
\begin{equation}
H_{LP}(\Pi _{id})=\int_{T^{\ast }\mathcal{Q}}\left\langle \Pi _{id}\left( 
\mathbf{z}\right) ,X_{h_{f}}\left( \mathbf{z}\right) \right\rangle \text{ }%
d\mu \left( \mathbf{z}\right)   \label{hlppi}
\end{equation}%
which is equivalent to the functional $H_{LP}(f)$ under the identification (%
\ref{density}). 

\begin{lemma}
For the Hamiltonian functional in Eq.(\ref{hlppi}) we have%
\begin{equation*}
\frac{\delta H_{LP}(\Pi _{id})}{\delta \Pi _{id}}=X_{h}(\mathbf{z}).
\end{equation*}
\end{lemma}

\section{Momentum Formulation of Dynamics}

\subsection{Lie-Poisson dynamics}

We apply the standard Lie-Poisson reduction to $G$ \cite{mw82},\cite{mar82},%
\cite{mw74}. The right invariant extensions to $T^{\ast }G$ of elements of $%
\mathfrak{g}^{\ast }$ are obtained through composition with the momentum map 
$\mathbb{J}_{L}(\Pi _{\varphi })=\Pi _{id}$. In particular, for a functional 
$H:\mathfrak{g}^{\ast }\rightarrow 
\mathbb{R}
$ the right invariant extension is the functional $H^{R}:T^{\ast
}G\rightarrow 
\mathbb{R}
$ defined by%
\begin{equation*}
H^{R}=H\circ \mathbb{J}_{L}\;,\;\;\;H^{R}\left( \varphi ,\Pi _{\varphi
}\right) =H(\Pi _{id})=H\left( \Pi _{\varphi }\circ \varphi ^{-1}\right) 
\end{equation*}%
and applying the chain rule with $\Pi _{\varphi }=\Pi _{id}\circ \varphi $
we have the differential 
\begin{equation*}
\delta H^{R}=\frac{\delta H^{R}}{\delta \varphi }\delta \varphi +\frac{%
\delta H^{R}}{\delta \Pi _{\varphi }}\delta \Pi _{\varphi }=\frac{\delta H}{%
\delta \Pi _{id}}\frac{\delta \Pi _{id}}{\delta \varphi }\delta \varphi +%
\frac{\delta H}{\delta \Pi _{id}}\frac{\delta \Pi _{id}}{\delta \Pi
_{\varphi }}\delta \Pi _{\varphi }.
\end{equation*}%
If $\Pi _{id}(\mathbf{z})=\Pi _{a}(\mathbf{z})dz^{a}=\Pi _{i}(\mathbf{z}%
)dq^{i}+\Pi ^{i}(\mathbf{z})dp_{i}$ we have $(\delta \Pi _{id}/\delta \Pi
_{\varphi })\mid _{id}=1$ and

\begin{equation*}
\text{\ }\dfrac{\delta \Pi _{id}}{\delta \varphi }\mid _{id}=\dfrac{\delta }{%
\delta \varphi }\Pi _{a}(\mathbf{z})dz^{a}\mid _{id}=-d\mathbf{\Pi }_{id}(%
\mathbf{z})
\end{equation*}%
with $\mathbf{\Pi }_{id}$ denoting the components of the one-form $\Pi _{id}$%
. Thus, for the differentiation of the right-invariant functionals we find

\begin{equation}
\frac{\delta H^{R}}{\delta \Pi _{\varphi }}|_{id_{G}}=\frac{\delta H}{\delta
\Pi _{id}}\;,\;\;\;\;\frac{\delta H^{R}}{\delta \varphi }|_{id_{G}}=-%
\overrightarrow{(\frac{\delta H}{\delta \Pi _{id}}}\cdot \nabla _{z})\mathbf{%
\Pi }_{id}
\end{equation}%
where the overhead arrow denotes the components of Lie algebra element. One
can now evaluate the canonical Poisson bracket on $T^{\ast }G$ at the
identity using the above relations. This gives $(+)$Lie-Poisson bracket on $%
\mathfrak{g}^{\ast }$, that is, $\{F^{R},G^{R}\}_{T^{\ast }G}\mapsto
\{F,G\}_{+LP}$. 

\begin{proposition}
Let $\Pi _{id}\in \mathfrak{g}^{\ast }$ and $\left[ \text{ },\text{ }\right] 
$ be the Jacobi-Lie bracket on $\mathfrak{g}$. Then the Lie-Poisson bracket
on $\mathfrak{g}^{\ast }$ is given by%
\begin{equation}
\left\{ H(\Pi _{id}),K(\Pi _{id})\right\} _{LP}=\int_{T^{\ast }\mathcal{Q}%
}\Pi _{id}(\mathbf{z})\cdot \left[ \frac{\delta H}{\delta \Pi _{id}(\mathbf{z%
})},\frac{\delta K}{\delta \Pi _{id}(\mathbf{z})}\right] \text{ }d\mu (%
\mathbf{z})  \label{lppi}
\end{equation}%
where $\delta H/\delta \Pi _{id}$ is regarded to be an element of $\mathfrak{%
g}$.
\end{proposition}

The derivation can also be found in the internet supplement to \cite{mr94}.
The Lie-Poisson dynamics on $\mathfrak{g}^{\ast }$ may be written as follows:

\begin{proposition}
The Hamiltonian vector fields on $\mathfrak{g}^{\ast }$ for the Lie-Poisson
structure defined by the bracket in Eq.(\ref{lppi}) have the form%
\begin{equation}
\frac{d\Pi _{id}}{dt}=-ad_{\delta H/\delta \Pi _{id}}^{\ast }(\Pi _{id})=%
\mathcal{L}_{\delta H/\delta \Pi _{id}}(\Pi _{id})=J_{LP}(\Pi _{id})\frac{%
\delta H}{\delta \Pi _{id}}  \label{pieq}
\end{equation}%
where the Hamiltonian operator defining the bracket (\ref{lppi}) is given by%
\begin{equation}
J_{LP}(\Pi _{id})=\left( 
\begin{array}{cc}
\Pi _{i}\frac{{\LARGE d}}{{\LARGE dq}^{j}}+\frac{{\Large d}}{{\Large dq}^{i}}%
\cdot \Pi _{j} & \text{ \ \ \ \ }\Pi ^{i}\frac{{\Large d}}{{\Large dq}^{j}}+%
\frac{{\Large d}}{{\Large dp}_{i}}\cdot \Pi _{j} \\ 
\Pi _{i}\frac{{\Large d}}{{\Large dp}_{j}}+\frac{{\Large d}}{{\Large dq}^{i}}%
\cdot \Pi ^{j} & \text{ \ \ \ \ }\Pi ^{i}\frac{{\Large d}}{{\Large dp}_{j}}+%
\frac{{\large d}}{{\Large dp}_{i}}\cdot \Pi ^{j}%
\end{array}%
\right)  \label{lpopi}
\end{equation}%
with $\frac{{\Large d}}{{\Large dq}^{i}}\cdot \Pi _{j}=\frac{{\Large d\Pi }%
_{j}}{{\Large dq}^{i}}+\Pi _{j}\frac{{\Large d}}{{\Large dq}^{i}}$ etc.
\end{proposition}

The operator in Eq.(\ref{lpopi}) is appearently skew adjoint with respect to
the $L^{2}$-norm and satisfies the Jacobi identity by construction \cite%
{mor98},\cite{olv86}. The relation of $J_{LP}(\Pi _{id})$ to the Lie-Poisson
bracket in Eq.$\left( \ref{lppi}\right) $ is

\begin{equation*}
\Pi _{id}\cdot \left[ \frac{\delta H}{\delta \Pi _{id}},\frac{\delta K}{%
\delta \Pi _{id}}\right] =\frac{\delta H}{\delta \Pi _{id}}\cdot J(\Pi
_{id})(\frac{\delta K}{\delta \Pi _{id}})-\nabla _{z}\cdot \frac{\delta K}{%
\delta \Pi _{id}}(\Pi _{id}\cdot \frac{\delta H}{\delta \Pi _{id}})
\end{equation*}%
where the divergence term on the right disappears upon integration.

\begin{remark}
The Hamiltonian operator $J_{LP}(\Pi _{id})$ may be considered to be a map
taking a Lie algebra element $X_{h\text{ }}$in $\mathfrak{g}$\ to the
corresponding generator $ad_{X_{h\text{ }}}^{\ast }$of the coadjoint action
of $\mathfrak{g}$ on $\mathfrak{g}^{\ast }$. That is,%
\begin{equation*}
J_{LP}(\Pi _{id}):\mathfrak{g}\longrightarrow (-ad_{\mathfrak{g}}^{\ast }:%
\mathfrak{g}^{\ast }\rightarrow \mathfrak{g}^{\ast })
\end{equation*}%
and with reference to particle phase space we have%
\begin{equation*}
J_{LP}:TT^{\ast }\mathcal{Q}\longrightarrow T\mathcal{O}\subseteq T\mathfrak{%
g}^{\ast }\subseteq TT^{\ast }T^{\ast }\mathcal{Q}
\end{equation*}%
where $\mathcal{O}(\Pi _{id})=\left\{ Ad_{\varphi ^{-1}}^{\ast }(\Pi
_{id})=\varphi _{\ast }(\Pi _{id})\text{ }\mid \text{ }\varphi \in G\right\} 
$ is the coadjoint orbit through $\Pi _{id}\in \mathfrak{g}^{\ast }$. \ 
\end{remark}

\subsection{\label{mvlasov}Vlasov equations in momentum variables}

In the Lie-Poisson setting, the Poisson-Vlasov equations arise from the
Hamiltonian functional that generates the kinematical symmetries. This is
the momentum function defined by means of the momentum map%
\begin{equation}
\mathbb{J}_{L}\left( X_{k}\right) \left( \Pi _{id}\right) =\int_{T^{\ast }%
\mathcal{Q}}\left\langle \Pi _{id}\left( \mathbf{z}\right) ,X_{k}\left( 
\mathbf{z}\right) \right\rangle d\mu \left( \mathbf{z}\right)   \label{xk}
\end{equation}%
for the lifted left action of $G.$ Here the generator $X_{k}\in \mathfrak{g}$
is yet to be specified. Due to the constraint imposed by the Poisson
equation, that is the constraint variational derivative, the relevant
Hamiltonian vector field $X_{k}$ in Eq.(\ref{xk}) turns out to be the one
associated with the function $h_{f}$ given in Eq.(\ref{hf}).

\begin{proposition}
For the right invariant Hamiltonian functional in Eq.(\ref{hlppi}) the
Lie-Poisson equations on $\mathfrak{g}^{\ast }$ are%
\begin{equation}
\frac{d\Pi _{i}\left( \mathbf{z}\right) }{dt}=-X_{h}(\Pi _{i}\left( \mathbf{z%
}\right) )+e\frac{\partial ^{2}\phi _{f}\left( \mathbf{q}\right) }{\partial
q^{i}\partial q^{j}}\Pi ^{j}\left( \mathbf{z}\right)   \label{momvlaa}
\end{equation}%
\begin{equation}
\frac{d\Pi ^{i}\left( \mathbf{z}\right) }{dt}=-X_{h}(\Pi ^{i}\left( \mathbf{z%
}\right) )-\frac{1}{m}\delta ^{ij}\Pi _{j}\left( \mathbf{z}\right) \text{ \ }
\label{momvla}
\end{equation}%
with the constraint 
\begin{equation}
\nabla _{q}^{2}\phi _{\mathbf{\Pi }}\left( \mathbf{q}\right) =e\int \nabla
_{q}\cdot \mathbf{\Pi }_{p}\left( \mathbf{z}\right) \text{ }d^{3}\mathbf{p.}
\label{mompoi}
\end{equation}
\end{proposition}

Eqs.(\ref{momvlaa}) and (\ref{momvla}), which we call the momentum-Vlasov
equations, follow from the Lie-Poisson structure expressed in coordinates of
the momentum one-form $\Pi _{id}$ and will be shown, in the next section, to
give rise to the Vlasov equation in the density variable. The proof of
proposition follows from the constraint variational derivative of $%
H_{LP}(\Pi _{id})$ obtained before and from Eq.(\ref{pieq}) by computing,
for example, the Lie derivative of $\Pi _{id}$ with respect to $X_{h}$.

\section{Density Formulation of Dynamics}

The definition of plasma density is motivated by the following observation.
Regarding the variational derivative of functions on $\mathfrak{g}^{\ast }$
as elements of $\mathfrak{g}$ means that there are functions $h,k$ on $%
T^{\ast }\mathcal{Q}$ such that the Jacobi-Lie bracket in Eq.(\ref{lppi})
can be written as canonical Poisson bracket of $h,k$. Then the Lie-Poisson
bracket on $\mathfrak{g}^{\ast }$ becomes 
\begin{equation}
\int_{T^{\ast }\mathcal{Q}}\,(\frac{\partial \Pi _{i}(\mathbf{z})}{\partial
p_{i}}-\frac{\partial \Pi ^{i}(\mathbf{z})}{\partial q^{i}})\,\{h(\mathbf{z}%
),k(\mathbf{z})\}_{T^{\ast }\mathcal{Q}}\,d\mu (\mathbf{z})
\end{equation}%
which requires, as in the definition of dual algebra, the divergence of the
vector $\Pi _{id}^{\sharp }$ to be non-zero.

\subsection{Introducing the plasma density function}

We first show that the definition in Eq.(\ref{density}) of plasma density
function leads to the correct Lie-Poisson structure in this variable. Then
we show that the relation between formulations of plasma dynamics in the
momentum variables $\Pi _{id}$ and the plasma density function $f$ can be
made precise in terms of a momentum map. We shall remark, in the next
section, by presenting a canonical Hamiltonian system in momentum variables,
that Eqs.(\ref{momvlaa}) and (\ref{momvla}) are not unique in the sense that
they yield the Vlasov equation in density variable.

\begin{proposition}
The Hamiltonian operator $J_{LP}(\Pi _{id})$ on $\mathfrak{g}^{\ast }$
transforms into the Hamiltonian operator%
\begin{equation}
J_{LP}(f)=\nabla _{p}f\cdot \nabla _{q}-\nabla _{q}f\cdot \nabla _{p}
\label{jlpf}
\end{equation}%
on the space $Den(T^{\ast }\mathcal{Q})$ of densities under the
correspondence in Eq.(\ref{density}).
\end{proposition}

\begin{proof}
Regarding the definition (\ref{density}) as a transformation of Eulerian
variables we can compute the tranformation of the Hamiltonian operator $%
J_{LP}(\Pi _{id})$ as follows. The derivative of $f$ in the direction of $%
\Pi _{id}$ is a $3\times 6$ matrix of differential operators 
\begin{equation}
D_{f}(\Pi _{id}(\mathbf{z}))=[(\frac{d}{dp_{j}})\;\;\;\;-(\frac{d}{dq^{j}}%
)]=[\nabla _{p}\text{ \ }-\nabla _{q}]
\end{equation}%
which transforms the Hamiltonian operator $J_{LP}(\Pi _{id})$ according to 
\begin{equation}
J_{LP}(f)=D_{f}(\Pi _{id})\cdot J_{LP}(\Pi _{id})\cdot D_{f}^{\ast }(\Pi
_{id})  \label{trham}
\end{equation}%
where $D_{f}^{\ast }$ is the adjoint of $D_{f}$ with respect to the $L^{2}-$%
norm \cite{olv86}. A direct computation from Eq.(\ref{trham}) gives the
operator $J_{LP}(f)$. \vspace{2mm}
\end{proof}

\begin{remark}
$D_{f}^{\ast }$ transforms the Lie algebra elements, that is if $H$ is a
functional on $Den(T^{\ast }\mathcal{Q})$, then $\delta H/\delta f$ is a
function on $T^{\ast }\mathcal{Q}$ and $D_{f}^{\ast }(\delta H/\delta f)$
gives the components of the Hamiltonian vector field in $\mathfrak{g}$
corresponding to the Hamiltonian function $\delta H/\delta f$.
\end{remark}

\begin{remark}
It can also be verified that the momentum-Vlasov equations (\ref{momvlaa})
and (\ref{momvla}) yield the Vlasov equation (\ref{vlasov}) for $f$ from the
definition of the density.
\end{remark}

We recall that the Lie-Poisson structure of the Vlasov equation in density
variable is defined by the bracket%
\begin{equation}
\left\{ H(f),K(f)\right\} _{LP}=\int_{T^{\ast }\mathcal{Q}}f(\mathbf{z})%
\text{ }\left\{ \frac{\delta H}{\delta f(\mathbf{z})},\frac{\delta K}{\delta
f(\mathbf{z})}\right\} _{T^{\ast }\mathcal{Q}}\text{ }d\mu (\mathbf{z})
\end{equation}%
associated with the operator $J_{LP}(f)$ and the Hamiltonian function $%
H_{LP}(f)$ \cite{mor80}.

\subsection{Momentum map description of density}

\begin{proposition}
The momentum map for the action of the additive group $\mathcal{F}(T^{\ast }%
\mathcal{Q})$ of functions on $\mathfrak{g}^{\ast }$ defines the plasma
density.
\end{proposition}

\begin{proof}
By definition in Eq.(\ref{gst}) of $\mathfrak{g}^{\ast }$, the one-form $\Pi
_{id}$ is invariant under the addition of an exact one-form on $T^{\ast }%
\mathcal{Q}$. So, $\mathcal{F}(T^{\ast }\mathcal{Q})$ is the gauge group in
the definition of $\mathfrak{g}^{\ast }$. With reference to the
identification of $\mathfrak{g}$ with $\mathcal{F}(T^{\ast }\mathcal{Q})$,
the action of $\mathcal{F}(T^{\ast }\mathcal{Q})$ on $\mathfrak{g}^{\ast }$
by translation can be regarded as the action of $\mathfrak{g}$ on $\mathfrak{%
g}^{\ast }$ by 
\begin{equation}
(X_{k},\Pi _{id})\mapsto \Pi _{id}+\Omega _{T^{\ast }\mathcal{Q}}^{\flat
}(X_{k})  \label{action}
\end{equation}%
where $\Omega _{T^{\ast }\mathcal{Q}}^{\flat }=(\Omega _{T^{\ast }\mathcal{Q}%
}^{\sharp })^{-1}$. This can be interpreted as the action of the underlying
vector space of $\mathfrak{g}$. Thus, we consider the Lie algebra
isomorphism $\mathcal{F}(T^{\ast }\mathcal{Q})\rightarrow \mathfrak{g}%
\mathbf{:}k\mathbf{\longrightarrow }X_{k}$ for the gauge equivalent classes
of one-forms in $\mathfrak{g}^{\ast }$. The dual of this is the required
momentum map 
\begin{equation}
\mathbb{J}_{tr}:\mathfrak{g}^{\ast }\rightarrow \mathcal{F}^{\ast }(T^{\ast }%
\mathcal{Q})\equiv Den(T^{\ast }\mathcal{Q})
\end{equation}%
for the definition of the plasma density from the momentum variables. From
definitions, we have%
\begin{equation}
\left\langle \mathbb{J}_{tr}(\Pi _{id}),k\right\rangle =\left\langle \Pi
_{id},X_{k}\right\rangle =\left\langle \nabla _{z}\circ \Omega _{T^{\ast }%
\mathcal{Q}}^{\sharp }\circ \Pi _{id},k\right\rangle   \label{57}
\end{equation}%
where $\nabla _{z}$ is taken to be the dual of the exterior derivative $d.$
When evaluated in Eulerian coordinates $\mathbf{z}$ the momentum map (\ref%
{57}) gives exactly the definition (\ref{density}) of the plasma density
function.
\end{proof}

\begin{remark}
It has been argued that the physical initial conditions must satisfy $f(%
\mathbf{z},0)>0$ \cite{mor81},\cite{mor00}. This requires the description of
density by elements $\Pi _{id}\in \mathfrak{g}^{\ast }$ with $\nabla
_{z}\cdot \Pi _{id}^{\sharp }(\mathbf{z})>0$. Equivalently, in the language
of differential forms, we have $d(\Pi _{id}\wedge \Omega _{T^{\ast }\mathcal{%
Q}}^{2})>0$. Consider a six dimensional domain $\mathit{D}$ in $T^{\ast }%
\mathcal{Q}$ with boundary $\partial \mathit{D}$. Then, the positive
divergence implies 
\begin{equation}
\int_{\partial \mathit{D}}\Pi _{id}(\mathbf{z})\wedge \Omega _{T^{\ast }%
\mathcal{Q}}^{2}(\mathbf{z})>0  \label{sscs}
\end{equation}%
so that we have a volume element or an orientation for the five dimensional
boundary of the region $\mathit{D}$. This can now be related to the
nondegeneracy of the orbit symplectic structure on $\mathfrak{g}^{\ast }$.
An element of the orbit through $\Pi _{id}$ will be of the form $\mathcal{L}%
_{X_{k}}(\Pi _{id})$. By definition, the orbit symplectic structure is 
\begin{equation*}
\Omega _{\Pi _{id}}\left( \mathcal{L}_{X_{k}}(\Pi _{id}),\mathcal{L}%
_{X_{g}}(\Pi _{id})\right) =\int_{\partial \mathit{D}}\{g(\mathbf{z}),k(%
\mathbf{z})\}\text{ }\Pi _{id}(\mathbf{z})\wedge \text{ }\Omega _{T^{\ast }%
\mathcal{Q}}^{2}(\mathbf{z})
\end{equation*}%
which, by Eq.(\ref{sscs}) does not vanish for arbitrary functions $g$ and $k$%
.
\end{remark}

\section{Equivalence of Momentum and Density Formulations}

\begin{proposition}
$H_{LP}(f)=H_{LP}\left( \Pi _{id}\right) $
\end{proposition}

\begin{proof}
Replace $f$ in Eq.(\ref{ede}) by its definition to get 
\begin{eqnarray*}
H_{LP}(f) &=&\int_{T^{\ast }\mathcal{Q}}\;{\frac{1}{2m}}p^{2}(\nabla
_{p}\cdot \mathbf{\Pi }_{q}(\mathbf{z})-\nabla _{q}\cdot \mathbf{\Pi }_{p}(%
\mathbf{z}))\text{ }d\mu (\mathbf{z}) \\
&&+{\frac{e^{2}}{2}}\iint\limits_{T^{\ast }Q}\nabla _{q}\cdot \mathbf{\Pi }%
_{p}(\mathbf{z})K(\mathbf{q}|\mathbf{\acute{q}})\nabla _{q^{\prime }}\cdot 
\mathbf{\Pi }_{p^{\prime }}(\mathbf{\acute{z}})\text{ }d\mu (\mathbf{z})d\mu
(\mathbf{\acute{z}})
\end{eqnarray*}%
upto divergence. The first and the second integrals are equivalent to 
\begin{equation}
-{\frac{\mathbf{p}}{m}}\cdot \mathbf{\Pi }_{q}(\mathbf{z})\text{ ,\ \ \ }-{%
\frac{e^{2}}{2}}\mathbf{\Pi }_{p}(\mathbf{z})\cdot \nabla _{q}(\;K(\mathbf{q}%
|\mathbf{\acute{q}})\nabla _{q^{\prime }}\cdot \mathbf{\Pi }_{p^{\prime }}(%
\mathbf{\acute{z}})\;)\text{,}
\end{equation}%
respectively. Using Green's function solution we obtain $H_{LP}(\Pi _{id})$.
Conversely, starting from the function $H_{LP}\left( \Pi _{id}\right) $ we
compute%
\begin{eqnarray*}
H_{LP}\left( \Pi _{id}\right)  &=&\int_{T^{\ast }\mathcal{Q}}(-\nabla
_{p}h_{f}(\mathbf{z})\cdot \mathbf{\Pi }_{q}(\mathbf{z})+\nabla _{q}h_{f}(%
\mathbf{z})\cdot \mathbf{\Pi }_{p}(\mathbf{z}))\text{ }d\mu \left( \mathbf{z}%
\right)  \\
&=&\int_{T^{\ast }\mathcal{Q}}h_{f}\left( \mathbf{z}\right) f\left( \mathbf{z%
}\right) \text{ }d\mu \left( \mathbf{z}\right) 
\end{eqnarray*}%
which verifies the equivalence of the Hamiltonian functionals of the
Lie-Poisson structures in momentum and density formulations
\end{proof}

\begin{proposition}
\label{cano}The canonical Hamiltonian system with respect to the symplectic
two form%
\begin{equation}
\omega \left( \Pi _{i},\Pi ^{i}\right) =\int_{T^{\ast }\mathcal{Q}}\delta
\Pi _{i}\left( \mathbf{z}\right) \wedge \delta \Pi ^{i}\left( \mathbf{z}%
\right) \text{ }d\mu \left( \mathbf{z}\right)  \label{pisymp}
\end{equation}%
and for the Hamiltonian functional 
\begin{equation}
H_{0}(\Pi _{id})=\int_{T^{\ast }\mathcal{Q}}\left( \Pi _{i}X_{h}(\Pi ^{i})+%
\frac{1}{2m}\delta ^{ij}\Pi _{i}\Pi _{j}+\frac{e}{2}\frac{\partial ^{2}\phi
_{f}}{\partial q^{i}\partial q^{j}}\Pi ^{i}\Pi ^{j}\right) \left( \mathbf{z}%
\right) \text{\textbf{\ }}d\mu \left( \mathbf{z}\right)  \label{quadham}
\end{equation}%
which is quadratic in the fiber coordinates of $T^{\ast }T^{\ast }\mathcal{Q}
$ gives the Vlasov equation.

\begin{proof}
First of all, the density $\mathcal{H}_{0}(\mathbf{z})$ of the Hamiltonian
functional $H_{0}$ satisfies the divergence equation 
\begin{equation}
{\frac{\partial \mathcal{H}_{0}(\mathbf{z})}{\partial t}}-\nabla _{z}\cdot
\left( {(}\Pi _{i}(X_{h}(\Pi ^{i})+\frac{1}{m}\delta ^{ij}\Pi _{j}))\mathbf{X%
}_{h}\right) (\mathbf{z})=0
\end{equation}%
which is the conservation law in Eulerian form. To obtain the canonical
equations for the Hamiltonian functional in Eq.(\ref{quadham}) we substitute
the expressions for $X_{h}$ and $\phi _{f}$ into Hamiltonian functional $%
H_{0}$ and rearrange the terms to obtain%
\begin{eqnarray*}
H_{0}\left( \Pi _{j},\Pi ^{j}\right)  &=&\int\limits_{T^{\ast }Q}\Pi
_{j}\left( \mathbf{z}\right) \left( {\frac{\delta ^{ik}p_{k}}{m}}\frac{%
\partial \Pi ^{j}}{\partial q^{i}}+\frac{1}{2m}\delta ^{ij}\Pi _{i}\right)
\left( \mathbf{z}\right) \text{ }d\mu \left( \mathbf{z}\right)  \\
&&-e^{2}\iint\limits_{T^{\ast }Q}\Pi _{j}\left( \mathbf{z}\right) \Pi
^{l}\left( \mathbf{\acute{z}}\right) \frac{\partial \Pi ^{j}\left( \mathbf{z}%
\right) }{\partial p_{i}}\frac{\partial ^{2}K\left( \mathbf{q},\mathbf{%
\acute{q}}\right) }{\partial q^{i}\partial \acute{q}^{l}}\text{ }d\mu \left( 
\mathbf{\acute{z}}\right) d\mu \left( \mathbf{z}\right)  \\
&&+\frac{e^{2}}{2}\iint\limits_{T^{\ast }Q}\Pi ^{l}\left( \mathbf{\acute{z}}%
\right) \Pi ^{i}\left( \mathbf{z}\right) \Pi ^{j}\left( \mathbf{z}\right) 
\frac{\partial ^{3}K\left( \mathbf{q},\mathbf{\acute{q}}\right) }{\partial
q^{i}\partial q^{j}\partial \acute{q}^{l}}\text{ }d\mu \left( \mathbf{\acute{%
z}}\right) d\mu \left( \mathbf{z}\right) \text{ .}
\end{eqnarray*}%
The variation with respect to the components $\Pi _{j}$ of $\mathbf{\Pi }_{q}
$ can easily be computed to give%
\begin{equation*}
\frac{\delta H_{0}}{\delta \Pi _{j}}\left( \mathbf{z}\right) =X_{h}\left(
\Pi ^{j}\left( \mathbf{z}\right) \right) +\frac{1}{m}\delta ^{ij}\Pi
_{i}\left( \mathbf{z}\right) =-\frac{d\Pi ^{j}\left( \mathbf{z}\right) }{dt}
\end{equation*}%
so that the first set of equations holds. For the other set of Hamilton's
equations, we first note that the second integral in $H_{0}$ may be written,
up to a divergence, as%
\begin{equation*}
e^{2}\iint\limits_{T^{\ast }Q}\frac{\partial \Pi _{j}\left( \mathbf{z}%
\right) }{\partial p_{i}}\Pi ^{l}\left( \mathbf{\acute{z}}\right) \Pi
^{j}\left( \mathbf{z}\right) \frac{\partial ^{2}K\left( \mathbf{q},\mathbf{%
\acute{q}}\right) }{\partial q^{i}\partial \acute{q}^{l}}\text{ }d\mu \left( 
\mathbf{\acute{z}}\right) d\mu \left( \mathbf{z}\right) 
\end{equation*}%
and the derivative of this with respect to $\Pi ^{j}\left( \mathbf{z}\right) 
$ gives%
\begin{equation}
e\frac{\partial \phi _{f}\left( \mathbf{q}\right) }{\partial q^{i}}\frac{%
\partial \Pi _{j}\left( \mathbf{z}\right) }{\partial p_{i}}+\frac{\partial }{%
\partial q^{j}}\left( e^{2}\int\limits_{T^{\ast }Q}\frac{\partial \Pi
_{k}\left( \mathbf{\acute{z}}\right) }{\partial \acute{p}_{i}}\Pi ^{k}\left( 
\mathbf{\acute{z}}\right) \frac{\partial K\left( \mathbf{q},\mathbf{\acute{q}%
}\right) }{\partial \acute{q}^{i}}\text{ }d\mu \left( \mathbf{\acute{z}}%
\right) \right) \text{.}  \label{int1}
\end{equation}%
Similarly, we compute the derivative of the last term in $H_{0}$\ with
respect to the components $\Pi ^{j}\left( \mathbf{z}\right) $ of $\mathbf{%
\Pi }_{p}$ to be%
\begin{equation}
e\Pi ^{i}\left( \mathbf{z}\right) \frac{\partial ^{2}\phi _{f}\left( \mathbf{%
q}\right) }{\partial q^{j}\partial q^{i}}+\frac{\partial }{\partial q^{j}}%
\left( \frac{e^{2}}{2}\int\limits_{T^{\ast }Q}\Pi ^{i}\left( \mathbf{\acute{z%
}}\right) \Pi ^{k}\left( \mathbf{\acute{z}}\right) \frac{\partial
^{2}K\left( \mathbf{q},\mathbf{\acute{q}}\right) }{\partial \acute{q}%
^{i}\partial \acute{q}^{k}}\text{ }d\mu \left( \mathbf{\acute{z}}\right)
\right) \text{.}  \label{int2}
\end{equation}%
Collecting these results we find%
\begin{eqnarray*}
\frac{\delta H_{0}}{\delta \Pi ^{j}}\left( \mathbf{z}\right) 
&=&-X_{h}\left( \Pi _{j}\left( \mathbf{z}\right) \right) +e\frac{\partial
^{2}\phi _{f}\left( \mathbf{q}\right) }{\partial q^{j}\partial q^{i}}\Pi
^{i}\left( \mathbf{z}\right) +\frac{\partial \Phi \left( \mathbf{q}\right) }{%
\partial q^{j}} \\
&=&\frac{d\Pi _{j}\left( \mathbf{z}\right) }{dt}+\frac{\partial \Phi \left( 
\mathbf{q}\right) }{\partial q^{j}}
\end{eqnarray*}%
where $\Phi \left( \mathbf{q}\right) $ is the function consisting of two
gradient terms in Eqs.(\ref{int1}) and (\ref{int2}). In passing to the
density formulation the additional gradient term is acted on by $\mathbf{p}$%
-divergence which identically vanishes because it is a function of $\mathbf{q%
}$ only. Thus, the Vlasov equations resulting from the canonical flow of the
quadratic Hamiltonian in Eq.(\ref{quadham}) and the momentum Vlasov
equations are the same.
\end{proof}
\end{proposition}

The canonical Hamiltonian system in proposition (\ref{cano}) can be regarded
as a manifestation of the fact that, the addition of an exact one-form to $%
\Pi _{id}$ does not affect the Vlasov equation in $f$. In fact, there are
infinitely many such flows in momentum formulation due to symmetries in the
definition of momentum variables. To this end, we observe the following
symmetries associated with the momentum formulations. The quadratic
Hamiltonian $H_{0}(\Pi _{id})$ is invariant under shifts in velocity
together with a reflection in position 
\begin{equation}
\Pi _{i}\mapsto \Pi _{i}+2mX_{h}(\Pi _{i})\;,\;\;\;\Pi ^{i}\mapsto -\Pi ^{i}.
\end{equation}%
Both the Hamiltonian and the canonical symplectic structure are invariant
under arbitrary diffeomorphisms in the variables $\Pi ^{i}$. The definition (%
\ref{density}) of the plasma density $f$ is also invariant under 
\begin{equation}
\Pi _{i}\mapsto \Pi _{i}+(\nabla _{p}\times \mathbf{A}_{p})_{i}\;,\;\;\;\;%
\Pi ^{i}\mapsto \Pi ^{i}+(\nabla _{q}\times \mathbf{A}_{q})^{i}
\end{equation}%
for arbitrary vector functions $\mathbf{A}_{q}(\mathbf{z})$ and $\mathbf{A}%
_{p}(\mathbf{z})$.

Regarding the components $\Pi _{i}$ of $\mathbf{\Pi }_{q}$ as \textit{%
momentum }variables in the canonical structure of proposition (\ref{cano}),
Eqs.(\ref{momvla}) becomes inverse Legendre transformations to be solved for
the \textit{momenta}. Then, the Lagrangian functional%
\begin{equation*}
L_{0}\left[ \mathbf{\Pi }_{p}\right] =\int_{T^{\ast }\mathcal{Q}}\left( 
\frac{m}{2}|X_{h}(\mathbf{\Pi }_{p})+\frac{d\mathbf{\Pi }_{p}}{dt}|^{2}-%
\frac{e}{2}\frac{\partial ^{2}\phi _{f}}{\partial q^{i}\partial q^{j}}\Pi
^{i}\Pi ^{j}\right) \left( \mathbf{z}\right) \text{\textbf{\ }}d\mu \left( 
\mathbf{z}\right) 
\end{equation*}%
involving the velocity $d\mathbf{\Pi }_{p}/dt$ shifted by the term $-X_{h}(%
\mathbf{\Pi }_{p})$, gives the Euler-Lagrange equations%
\begin{equation*}
\ddot{\Pi}^{i}(\mathbf{z})+2X_{h}(\dot{\Pi}^{i}(\mathbf{z}))+X_{h}^{2}(\Pi
^{i}(\mathbf{z}))+\frac{e}{m}\delta ^{ij}\frac{\partial ^{2}\phi _{f}(%
\mathbf{q})}{\partial q^{k}\partial q^{j}}\Pi ^{k}(\mathbf{z})=0
\end{equation*}%
which can also be obtained from Eqs.(\ref{momvlaa}) and (\ref{momvla}) by
eliminating the variables $\Pi _{i}$. 

\section{Comparison with 2D Incompressible Fluid}

The motion of an incompressible fluid in a two dimensional region $\mathcal{M%
}$ $\mathcal{\subseteq }$ $\mathcal{%
\mathbb{R}
}^{2}$ is described by volume (area) preserving diffeomorphisms. The
generators are divergence-free vectors. Since the dimension is two, these
are equivalent to canonical diffeomorphisms and the generators are canonical
Hamiltonian vectors of the form $\mathbf{v}=-X_{\psi }$ where the
Hamiltonian function $\psi $ is the stream function. The curl of $\mathbf{v}$
or, equivalently, the Laplacian of the stream function $\psi $ is the
vorticity. Geometrically, it is a two-form given by 
\begin{equation}
\omega =d\left( \mathbf{v}\cdot d\mathbf{z}\right) =\nabla _{z}^{2}\psi
dq\wedge dp,\text{ \ \ \ }\mathbf{z}=\left( q,p\right) 
\end{equation}%
or, as a function $\omega =\nabla _{z}\cdot \Omega _{M}^{\sharp }\left( 
\mathbf{v}\cdot d\mathbf{z}\right) $.

The dynamics is governed by the Euler equation in vorticity form 
\begin{equation*}
\frac{\partial \omega }{\partial t}=\left\{ \omega ,\psi \right\} _{\mathcal{%
M}}
\end{equation*}%
where $\left\{ \text{ },\text{ }\right\} _{\mathcal{M}}$ is the canonical
Poisson bracket on $\mathcal{M}.$ This is the Lie-Poisson equation on the
dual of the Lie algebra of the group of volume preserving diffeomorphisms
for the Lie-Poisson structure

\begin{equation}
\left\{ H(\omega ),K(\omega )\right\} _{LP}=\int_{\mathcal{M}}\omega (%
\mathbf{z})\left\{ \frac{\delta H}{\delta \omega (\mathbf{z})},\frac{\delta K%
}{\delta \omega (\mathbf{z})}\right\} _{\mathcal{M}}d^{2}\mathbf{z}
\end{equation}%
with the Hamiltonian functional%
\begin{equation}
H=\frac{1}{2}\int \mathbf{v}^{2}\text{ }d^{2}\mathbf{z}=\frac{1}{2}\int
\left( \nabla _{z}\psi \right) ^{2}\text{ }d^{2}\mathbf{z}=-\frac{1}{2}\int
\psi \omega \text{ }d^{2}\mathbf{z}
\end{equation}%
where a divergence term in the last expression is omitted. Alternatively,
defining the Green's function solution%
\begin{equation}
\psi \left( \mathbf{z}\right) =-\int K(\mathbf{z}|\mathbf{z}^{\prime
})\omega \left( \mathbf{z}^{\prime }\right) \text{ }d^{2}\mathbf{z}^{\prime }
\end{equation}%
for the equation $\omega =\nabla _{z}^{2}\psi $ we have 
\begin{equation*}
H=\frac{1}{2}\int \int \omega \left( \mathbf{z}\right) K(\mathbf{z}|\mathbf{z%
}^{\prime })\omega \left( \mathbf{z}^{\prime }\right) \text{ }d^{2}\mathbf{z}%
d^{2}\mathbf{z}^{\prime }.
\end{equation*}%
See references \cite{mor82},\cite{mor98} from which we extract the above
summary, and \cite{arn66}-\cite{mwxx} for more on fluid motions. 

We observe that the quadratic Hamiltonian functional for incompressible
fluid is a direct consequence of the definition of the dual of the Lie
algebra by a metric. In this case, the (weak) non-degeneracy of the pairing
is the same as the non-degeneracy of the metric and the Lie algebra can be
identified with its metric dual. On the other hand, the Lie algebra $%
\mathfrak{g}$ of Hamiltonian vector fields and its dual $\mathfrak{g}^{\ast }
$ are $L^{2}-$orthogonal in $TT^{\ast }\mathcal{Q}$.

The momentum-Vlasov equations in components of $\Pi _{id}$ expresses the
evolution of a volume cell, that is the density $f$, in the phase space $%
T^{\ast }\mathcal{Q}$ in terms of its boundaries, that is, surfaces of the
momenta $\Pi _{id}$. This interpretation was first given by Ye and Morrison
in \cite{ym92} for the Clebsch variables $(\alpha ,\beta )$ defined by $%
\left\{ \alpha ,\beta \right\} _{T^{\ast }\mathcal{Q}}=f$. In the present
context, they form a non-closed one-form $\alpha d\beta $ and can be
identified with $\Pi _{id}$.

\begin{equation*}
\begin{tabular}{|c|c|c|}
\hline
& $2D-Fluid$ & $1D-Plasma$ \\ \hline
$%
\begin{array}{c}
Configuration \\ 
space%
\end{array}%
$ & $Diff_{vol}(\mathcal{M})$ & $Diff_{can}(T^{\ast }\mathcal{Q})$ \\ \hline
$%
\begin{array}{c}
particle \\ 
motion%
\end{array}%
$ & $%
\begin{array}{c}
volume \\ 
preserving \\ 
diffeomorphisms%
\end{array}%
$ & $%
\begin{array}{c}
Hamiltonian \\ 
diffeomorphisms%
\end{array}%
$ \\ \hline
$%
\begin{array}{c}
Lie\text{ }algebra \\ 
(generators\text{ }of\text{ } \\ 
motion)%
\end{array}%
$ & $%
\begin{array}{c}
divergence-free \\ 
vector\;fields%
\end{array}%
$ & $%
\begin{array}{c}
Hamiltonian \\ 
vector\;fields%
\end{array}%
$ \\ \hline
$%
\begin{array}{l}
identification\text{ } \\ 
of\text{ }Lie\text{ } \\ 
algebra\text{ }with\; \\ 
functions%
\end{array}%
$ & $%
\begin{array}{c}
\psi :\;stream \\ 
functions%
\end{array}%
$ & $%
\begin{array}{c}
h:\;Hamiltonian \\ 
functions%
\end{array}%
$ \\ \hline
$%
\begin{array}{c}
dual\text{ }of \\ 
Lie\text{ }algebra%
\end{array}%
$ & $%
\begin{array}{c}
g^{\flat }(v):metric\text{ }dual \\ 
of\text{ }velocity%
\end{array}%
$ & $%
\begin{array}{c}
non-closed \\ 
one-forms%
\end{array}%
$ \\ \hline
$%
\begin{array}{l}
identification\text{ } \\ 
of\text{ }dual \\ 
with\text{ }function\text{ } \\ 
spaces%
\end{array}%
$ & $%
\begin{array}{c}
\omega =\nabla \circ \Omega _{\mathcal{M}}^{\sharp }\circ g^{\flat }(v) \\ 
=\nabla ^{2}\psi%
\end{array}%
$ & $f=\nabla \circ \Omega _{T^{\ast }\mathcal{Q}}^{\sharp }\circ \Pi _{id}$
\\ \hline
$L^{2}-dual$ & $two-forms$ & $%
\begin{array}{c}
non-closed \\ 
one-form\text{ }densities%
\end{array}%
$ \\ \hline
$%
\begin{array}{l}
Clebsch \\ 
variables%
\end{array}%
$ & $%
\begin{array}{c}
\begin{array}{l}
g^{\flat }(v)=\alpha d\beta \\ 
\omega =\left\{ \alpha ,\beta \right\}%
\end{array}%
\end{array}%
$ & $%
\begin{array}{c}
\begin{array}{l}
\Pi _{id}=\alpha d\beta \\ 
f=\left\{ \alpha ,\beta \right\}%
\end{array}%
\end{array}%
$ \\ \hline
$%
\begin{array}{c}
Hamiltonian \\ 
functionals%
\end{array}%
$ & $%
\begin{array}{c}
\begin{array}{r}
H=-{\frac{1}{2}}\int \omega (\mathbf{z})\psi (\mathbf{z})d^{2}\mathbf{z} \\ 
={\frac{1}{2}}\int v^{2}(\mathbf{z})d^{2}\mathbf{z}%
\end{array}%
\end{array}%
$ & $%
\begin{array}{c}
\begin{array}{r}
H_{LP}={\frac{1}{2}}\int f(\mathbf{z})h(\mathbf{z})d^{2}\mathbf{z} \\ 
={\frac{1}{2}}\int <X_{h},\Pi _{id}>d^{2}\mathbf{z}%
\end{array}%
\end{array}%
$ \\ \hline
$%
\begin{array}{c}
dynamical \\ 
equations%
\end{array}%
$ & $%
\begin{array}{c}
\begin{array}{r}
{\frac{\partial \omega }{\partial t}}=\{\psi ,\omega \} \\ 
Euler^{\prime }s\text{ }equation\text{ }in \\ 
vorticity\text{ }form%
\end{array}%
\end{array}%
$ & $%
\begin{array}{c}
\begin{array}{r}
{\frac{\partial f}{\partial t}}=\{h,f\} \\ 
Vlasov^{\prime }s\text{ }equation%
\end{array}%
\end{array}%
$ \\ \hline
$%
\begin{array}{c}
Poisson \\ 
equation%
\end{array}%
$ & $%
\begin{array}{c}
\begin{array}{c}
\nabla ^{2}\phi =\omega \\ 
\mathit{as\ }definition%
\end{array}%
\end{array}%
$ & $%
\begin{array}{c}
\begin{array}{c}
\nabla ^{2}\phi =-\int f(\mathbf{z})d^{3}\mathbf{p} \\ 
as\text{ }constraint%
\end{array}%
\end{array}%
$ \\ \hline
\end{tabular}%
\end{equation*}

\bigskip

\bigskip

Inspired from the relation $\omega =\nabla _{z}^{2}\psi $ between the
vorticity and Hamiltonian functions of 2D fluid, we can establish a similar
relation between the plasma density function $f$ and the Hamiltonian
function $h$ of particle motion. The Hessian of $h$ can be considered to be
a map $Hess(h):TT^{\ast }\mathcal{Q}\rightarrow T^{\ast }T^{\ast }\mathcal{Q}
$ which is non-degenerate if the potential function $\phi _{f}$ is
non-degenerate. Let $X\in TT^{\ast }\mathcal{Q}$, and define the vector field%
\begin{equation}
Y=\Omega _{T^{\ast }\mathcal{Q}}^{\sharp }\circ Hess(h)\circ X\text{ .}
\label{hessmap}
\end{equation}%
If $X$ is a Hamiltonian vector field with a Hamiltonian function which is at
least quadratic in momenta, then $Y$ is not Hamiltonian. In particular, if
we choose the Hamiltonian function to be $h$ and identify $Y$ with $\Pi
_{id}^{\sharp }$, then we get the Poisson equation. In other words, the
Poisson equation in plasma resembles the relation between vorticity and
stream functions of 2D fluid. This relation may also be described by
assuming a non-degenerate Lagrangian functional on $TT^{\ast }\mathcal{Q}$
which is yet to be found. Suppose we have a Lagrangian $l$ on $TT^{\ast }%
\mathcal{Q}$ quadratic in the velocities $(\mathbf{\dot{q}},\mathbf{\dot{p}})
$. We can introduce the momenta which reads%
\begin{equation*}
\Pi _{i}={\frac{\delta l}{\delta \dot{q}^{i}}}={\frac{1}{m}}%
p_{i}\;,\;\;\;\;\Pi ^{i}={\frac{\delta l}{\delta \dot{p}_{i}}}=-e\delta ^{ij}%
{\frac{\partial \phi _{f}}{\partial q^{j}}}
\end{equation*}%
for the special choice $h$ of the Hamiltonian function. Then, the definition
of plasma density in terms of momentum variables gives $f=1/m+e\nabla
_{q}^{2}\phi _{f}$. Note also that with a rescaling of $m$ and a
redefinition of $f$ we can write $f$ $=$ $tr(Hess(h)).$

\section{Conclusions}

Gauge symmetries of the Hamiltonian motion of the plasma particles leads to
the kinematical constraint described by the Poisson equation. Thus, the
Poisson type equations naturally arise in kinetic theories of particles
moving in accordance with a canonical Hamiltonian formulation. Moreover,
this implies that the true configuration space appropriate for the dynamical
formulation of the collisionless plasma motion in Eulerian variables is the
semi-direct product space $\mathcal{F}(\mathcal{Q})$ $\circledS $ $%
Diff_{can}(T^{\ast }\mathcal{Q})$. This is well suited for a geometric
understanding of the limit $c\mapsto \infty $ of the Maxwell-Vlasov
equations. 

The formulation of dynamics in density variable is obtained by further
reduction of momentum-Vlasov equations by the symmetry defining the gauge
equivalence classes of momentum variables. The gauge algebra is, as a vector
space, shown to be the same as $\mathfrak{g}$ but with an action different
from the coadjoint action. The Eulerian velocity and momenta are
complementary in the vector space $TT^{\ast }\mathcal{Q}$. Obviously, this
and other geometric properties disappear upon identification of $\mathfrak{g}
$ and $\mathfrak{g}^{\ast }$ with function spaces $\mathcal{F}(T^{\ast }%
\mathcal{Q})$ and $Den(T^{\ast }\mathcal{Q})$, respectively. As an example,
the function $h_{f}$ and the density $f$ appear symmetrically in the
Hamiltonian functional of the Lie-Poisson structure whereas the
corresponding variables $X_{h_{f}}$ and $\Pi _{id}$ are complementary in the
sense that $\Omega _{T^{\ast }\mathcal{Q}}^{\flat }(\mathfrak{g})$ and $%
\mathfrak{g}^{\ast }$ decompose the space of one-forms on $T^{\ast }\mathcal{%
Q}$ into spaces of exact and non-closed one-forms, respectively. The
momentum formulation clarifies the geometric relation between the motions of
plasma particles and the Lie-Poisson description of dynamics \cite{gpd2}. We
expect the space $TT^{\ast }\mathcal{Q}$ be important for Euler-Poincar\'{e}
formulation of dynamics \cite{poi01},\cite{poi01b},\cite{hmr98}, and for
application of Tulczyjew construction for Legendre transformation \cite{tulc}
from Lie-Poisson formulation. 

\section{Acknowledgement}

I am grateful to Jerry Marsden, Phil Morrison, Ahmet Aydemir and O\u{g}ul
Esen for many collaborations at various stages of this and other works on
geometry of plasma dynamics.

\end{document}